\documentclass[letterpaper, 10 pt, conference]{ieeeconf}  

\usepackage[utf8]{inputenc}
\usepackage{hyperref}       
\usepackage{url}            
\usepackage{booktabs}       
\usepackage{amsfonts}       
\usepackage{nicefrac}       
\usepackage{microtype}      
\usepackage{bbold}
\usepackage{capt-of}
\usepackage{tikz}

\usepackage[center]{caption}

\usepackage{amsmath, 
amssymb, graphicx, cite, epsfig,epstopdf,color,soul}
\usepackage{tabularx}
\usepackage[ruled,vlined]{algorithm2e}
\allowdisplaybreaks
\usepackage{color}
\usepackage{enumerate}
\usepackage{multicol}
\usepackage{empheq}
\usepackage{caption} 

\newcommand{\Rset}{\mathbb{R}}

\newtheorem{thm}{\textbf{Theorem}}

\newtheorem{remark}{\textbf{Remark}}

\newtheorem{theorem}{Theorem}[section]

\newtheorem{lemma}[theorem]{Lemma}

\newtheorem{assumption}{Assumption}

\IEEEoverridecommandlockouts                              

\title{\Large \bf Finite-Time Analysis of Projected Two-Time-Scale Stochastic Approximation
}

\author{Yitao Bai, Thinh T. Doan, Justin Romberg
\thanks{This work was supported in part by NSF under CAREER Award 2527059 and Grant CIF 2527044.}
\thanks{Yitao Bai and Thinh T. Doan are with Department of Aerospace Engineering,
        University of Texas at Austin
        {\tt\small Email: yb5429@my.utexas.edu, thinhdoan@utexas.edu}}%
\thanks{Justin Romberg is with Department of Electrical and Computer Engineering,
        Georgia Tech
        {\tt\small Email: jrom@ece.gatech.edu}}%
}

\begin{document}

\maketitle
\vspace{-0.4in}
\thispagestyle{empty}
\pagestyle{empty}

\begin{abstract}
We study the finite-time convergence of projected linear two-time-scale stochastic approximation with constant step sizes and Polyak--Ruppert averaging. We establish an explicit mean-square error bound, decomposing it into two interpretable components, an approximation error determined by the constrained subspace and  a statistical error decaying at a sublinear rate, with constants expressed through restricted stability margins and a coupling invertibility condition. These constants cleanly separate the effect of subspace choice (approximation errors) from the effect of the averaging horizon (statistical errors). We illustrate our theoretical results through a number of numerical experiments on both synthetic and reinforcement learning problems.

\end{abstract}
\vspace{-0.05in}

\section{Introduction}
\vspace{-0.05in}
Stochastic approximation (SA) is a foundational tool
for solving root-finding and optimization problems from
noisy observations, with applications spanning signal
processing, control, and machine learning
\cite{KushnerYin03}.
Two-time-scale SA (TTSA) generalizes the
single-variable setting to coupled systems where two
groups of variables are updated on different time
scales---one ``fast'' and one ``slow''---and has found
broad use in actor--critic reinforcement learning
\cite{KondaTsitsiklis04}, gradient temporal-difference
(GTD) methods for policy evaluation
\cite{SuttonEtAl09GTD}, primal--dual constrained
optimization \cite{NedicOzdaglar2009SaddlePoint},
and bilevel learning
\cite{HongWaiWangYang2020TTBilevel}.

In many of these applications, the ambient dimension of
the iterates is too large for direct computation, and
one must restrict the variables to low-dimensional
linear subspaces \cite{TsitsiklisVanRoy97,BertsekasYu09}.
A prominent example is policy evaluation in
reinforcement learning: in the tabular setting, GTD
\cite{SuttonEtAl09GTD} is a two-time-scale SA
operating on the full value vector
$V^\pi\in\mathbb R^{|\mathcal S|}$ and an auxiliary
vector $u\in\mathbb R^{|\mathcal S|}$; when the state
space is large, one imposes the linear function
approximation $V^\pi\approx\Phi\theta$
\cite{TsitsiklisVanRoy97} by projecting onto the
feature subspace $\mathrm{span}(\Phi)$.
This projection reduces dimensionality but introduces
an approximation bias: the constrained solution generally
differs from the true one, and the subspace dimension
must be traded off against the resulting approximation
error \cite{MouPanWai23}.
Understanding this approximation--estimation tradeoff
quantitatively is essential for principled subspace
selection.

\noindent\textbf{Main Contributions.}
In this paper, we study finite-time convergence properties of the projected linear TTSA methods with constant step sizes and Polyak--Ruppert (PR) averaging.
We establish an explicit mean-square error bound for the
PR averages, where we decompose the error into two terms:
(i)~an \textbf{approximation component} determined by
the subspace distances
$\inf_{x\in\mathcal X}\|x^*-x\|$ and
$\inf_{y\in\mathcal Y}\|y^*-y\|$; and (ii)~a
\textbf{statistical component} decaying at rate $O(1/T)$,
where $T$ is the number of iterations (see Figure~\ref{fig:decomposition}).
A key feature of our analysis is the use of constant
step sizes, which yield exact telescoping identities
and produce a clean decomposition whose constants are
interpretable through restricted stability margins and
a coupling invertibility condition.
We illustrate our theoretical results though a number of numerical experiments on both synthetic and reinforcement learning problems. Our simulations confirm the bias--variance decomposition and the decaying rate $O(1/T)$ of statistical errors studied in our theoretical results.
\vspace{-0.1in}

\begin{figure}[t]
\centering\includegraphics[width=0.85\linewidth]{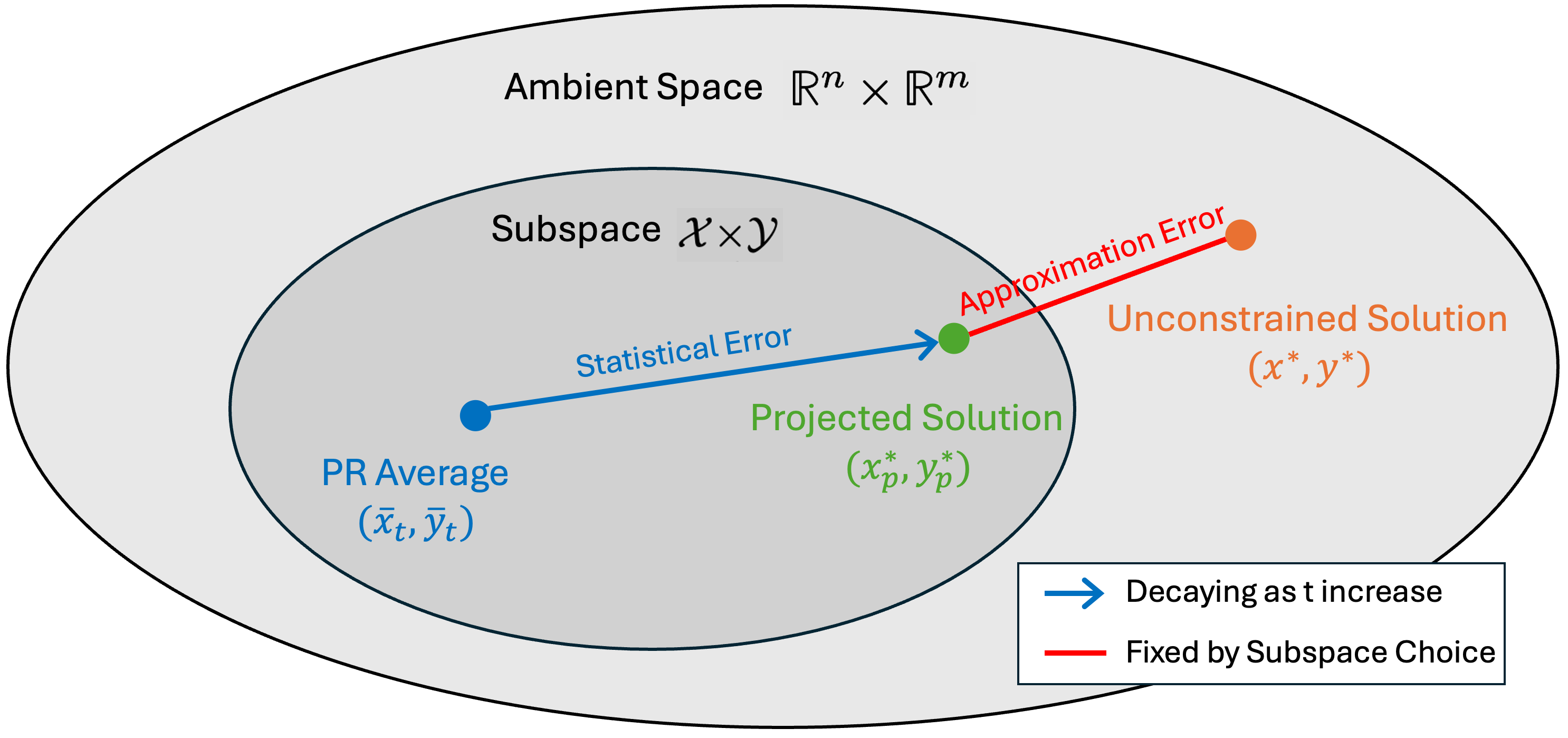}
  \vspace{-0.05in}
  \caption{Error decomposition of projected TTSA.}
  \label{fig:decomposition}
  \vspace{-0.25in}
\end{figure}

\subsection{Related Work}
\vspace{-0.05in}

\noindent\textbf{Projected single-time-scale SA.} Given its broad applicability, SA has been studied extensively in the literature~\cite{Borkar97,KushnerYin03}. The projected variant of SA has been studied in the context of approximate dynamic programming~\cite{TsitsiklisVanRoy97,BertsekasYu09,Bertsekas11}. In~\cite{MouPanWai23}, the authors establish oracle inequalities for PR averaging of the classical single-time-scale SA, obtaining a bias--variance decomposition that matches instance-dependent lower bounds. Our work, by contrast, studies the two-time-scale SA, which requires a fundamentally different analysis.

\noindent\textbf{Two-time-scale SA and PR averaging.}
The asymptotic convergence of TTSA is well understood via ODE methods~\cite{Borkar97,KushnerYin03}, with asymptotic rates established in~\cite{KondaTsitsiklis04}. Finite-time results have been obtained under various settings: \cite{Dalal18a} studies lock-in probability and projected TTSA for GTD; \cite{DoanRomberg20a} considers distributed TTSA; \cite{Kaledin20a} handles Markovian noise; and \cite{kong2025ttsa} establishes a nonasymptotic CLT. Separately, PR averaging~\cite{Ruppert88,Polyak90,PolyakJuditsky92} is known to improve the statistical efficiency of SA, with sharp single-time-scale results appearing in~\cite{MouLiWaiBarJor20}. However, none of these works combine PR averaging with subspace-constrained TTSA or provide an explicit bias--variance decomposition for the projected setting. This work addresses this gap by providing a finite-time theoretical analysis of PR averaging for iterates generated by projected TTSA.
  
\section{Projected Two-Time-Scale Stochastic Approximation} \label{sec:problem_formulation}

\vspace{-0.05in}
Two-time-scale stochastic approximation (TTSA) is an iterative method to find a pair $(x^*,y^*) \in \Rset^{n}\times\Rset^{m}$ satisfying  

\vspace{-0.05in}\noindent 
\begin{align}
\left\{ \begin{array}{l}
    g(x^*,y^*) =   A_{ff}\,x^* + A_{fs}\,y^*  - b_1 = 0, \\
h(x^*,y^*) =  A_{sf}\,x^* + A_{ss}\,y^* - b_2 = 0.
\end{array}\right. \label{eq:root}
\end{align}

\vspace{-0.05in}\noindent In particular, TTSA iteratively updates $(x_{k},y_{k})$, an estimate of $(x^*,y^*)$, using noisy samples of the system matrices as 
\vspace{-0.05in}\noindent 
\begin{align*}
    x_{t+1} &= x_{t} + \alpha(g(x_{t},y_{t}) + \epsilon_{t}),\\
    y_{t+1} &= x_{t} + \beta(g(x_{t},y_{t}) + \psi_{t}),
\end{align*}

\vspace{-0.05in}\noindent where $\beta\ll\alpha$ are two step sizes and the random variables $\epsilon,\psi$ represent sampling noises. As $\beta\ll\alpha$, $x$ is referred to as the fast variable while $y$ is called the slow variable \cite{Borkar97,KondaTsitsiklis04}.   

In this paper, we consider the setting where $(x,y)$ are constrained in some linear subspaces $(\mathcal{X},\mathcal{Y}) \subseteq \mathbb{R}^{d}\times\mathbb{R}^{r}$, where $r,d$ are much smaller than $m,n$. Our focus is to study a projected version of TTSA given as
\vspace{-0.05in}\noindent 
\begin{align}
    \begin{aligned}
        x_{t+1} &= \Pi_{\mathcal{X}}\left[x_{t} + \alpha(g(x_{t},y_{t}) + \epsilon_{t})\right],\\
    y_{t+1} &= \Pi_{\mathcal{Y}}\left[x_{t} + \beta(g(x_{t},y_{t}) + \psi_{t})\right],
    \end{aligned}\label{eq:alg}
\end{align}

\vspace{-0.05in}\noindent where $\Pi_{\mathcal{X}},\Pi_{\mathcal{Y}}$ denote the projections to the sets $\mathcal{X},\mathcal{Y}$, respectively. As $\mathcal{X},\mathcal{Y}$ are linear, we can obtain a closed form representation of these projection operators: let $U_{\mathcal X}\in\mathbb{R}^{n\times d}$ and
$U_{\mathcal Y}\in\mathbb{R}^{m\times r}$ have
orthonormal columns spanning $\mathcal X$ and
$\mathcal Y$, then the projection operators are defined as
\vspace{-0.05in}
\begin{equation*}
  \Pi_{\mathcal X}
    := U_{\mathcal X}\,U_{\mathcal X}^\top\quad\text{and}
  \quad
  \Pi_{\mathcal Y}
    := U_{\mathcal Y}\,U_{\mathcal Y}^\top.\vspace{-0.05in}
\end{equation*}

\vspace{-0.03in}While the constrained variant of TTSA arises naturally in many applications, including primal--dual constrained optimization \cite{NedicOzdaglar2009SaddlePoint}, bilevel learning \cite{HongWaiWangYang2020TTBilevel}, and low-rank parameterization in large-scale systems \cite{TsitsiklisVanRoy97}, our main motivation is the application of \eqref{eq:alg} to reinforcement learning, in particular, to study the gradient temporal-difference learning with linear function approximation (see Section~\ref{sec:simulation} for a detailed presentation).  

\noindent\textbf{Main focus.} The main objective of this paper is to study the convergence properties of the following PR averages of the iterates generated by \eqref{eq:alg} \vspace{-0.05in}\noindent 
\begin{align}
    \begin{aligned}
  \bar x_T := \frac{1}{T}\sum_{t=0}^{T-1} x_t,
  \quad
  \bar y_T := \frac{1}{T}\sum_{t=0}^{T-1} y_t.
\end{aligned}\label{eq:pr}
\end{align}

\vspace{-0.05in}\noindent PR averaging improves the statistical efficiency of constant-step-size SA by canceling the zero-mean fluctuations of the iterates around their limit, achieving an $O(1/T)$ mean-square error rate \cite{Polyak90,PolyakJuditsky92,MouLiWaiBarJor20}.

Note that due to the projections, one would not expect the sequence $(x_{t},y_{t})$ generated by \eqref{eq:alg}, if they converge, to reach $(x^*,y^*)$. Indeed, the iterates are confined to $\mathcal X\times\mathcal Y$, and the best one can hope for is convergence to the \emph{constrained solution} $(x_p^*,y_p^*)\in\mathcal X\times\mathcal Y$, defined as the unique pair satisfying the projected equations $\Pi_{\mathcal X}(A_{ff}x_p^* + A_{fs}y_p^* - b_1) = 0$ and $\Pi_{\mathcal Y}(A_{sf}x_p^* + A_{ss}y_p^* - b_2) = 0$. When $(x^*,y^*)\notin\mathcal X\times\mathcal Y$, the constrained solution differs from the true one, and the total error of the PR averages decomposes as
\vspace{-0.05in}
\begin{equation*}
\mathbb E\|\bar x_T - x^*\|^2
  \;\le\;
  \underbrace{2\|x_p^* - x^*\|^2}_{\text{approximation error}}
  \;+\;
  \underbrace{2\,\mathbb E\|\bar x_T - x_p^*\|^2}_{\text{statistical error}},
\end{equation*}

\vspace{-0.05in}\noindent and similarly for $y$. The \emph{approximation error} is the irreducible distance between the constrained and unconstrained solutions, determined entirely by the choice of subspaces. The \emph{statistical error} measures how well the PR averages track the constrained solution, and decays as $O(1/T)$. In the next subsection, we state the technical assumptions that guarantee existence and uniqueness of both $(x^*,y^*)$ and $(x_p^*,y_p^*)$, and that make the approximation error quantifiable. 
\vspace{-0.05in}\noindent 

\subsection{ Technical Assumptions}

\vspace{-0.05in}We next present the main assumptions underlying our analysis. Since the projections $\Pi_{\mathcal X},\Pi_{\mathcal Y}$ are linear, the constrained problem in \eqref{eq:alg} inherits a linear structure analogous to that of Eq. \eqref{eq:root}. Accordingly, our assumptions parallel the standard conditions in the TTSA literature, extended to account for the subspace restriction. Assumptions~\ref{as:block_stab} and~\ref{as:proj_inv} guarantee existence and uniqueness of  $(x^*,y^*)$ and
$(x_p^*,y_p^*)$, respectively.
Assumption~\ref{as:resolvent_smallgain} imposes a variant of standard well-posedness conditions on the resolvent matrices associated with the linear subspaces, which will be used to characterize the approximation errors 
$\|x_p^*-x^*\|^2$ and $\|y_p^*-y^*\|^2$. Finally, Assumption~\ref{as:md} requires that the sampling noises  $\epsilon_{t},\psi_{t}$ are Martingale difference with bounded variances. Assumptions~\ref{as:block_stab}--\ref{as:md} are fairly standard in the SA literature \cite{Borkar97,kong2025ttsa,MouPanWai23}.  

\begin{assumption}\label{as:block_stab}
$A_{ff}$ and the Schur complement
$\Delta := A_{ss} - A_{sf}\,A_{ff}^{-1}\,A_{fs}$
are Hurwitz. 
\end{assumption}


\begin{assumption}\label{as:proj_inv}
Let
$c_1 := U_{\mathcal X}
  (U_{\mathcal X}^\top A_{ff}\,U_{\mathcal X})^{-1}
  U_{\mathcal X}^\top$.
The projected system matrices
$U_{\mathcal X}^\top A_{ff}\,U_{\mathcal X}$ and
$U_{\mathcal Y}^\top
  (A_{ss}+A_{sf}\,c_1\,A_{fs})\,U_{\mathcal Y}$
are Hurwitz.
\end{assumption}\vspace{0.1cm}

For convenience, we consider the following notation

\vspace{-0.05in}\noindent 
\begin{align}\label{eq:c3c4}
  \begin{aligned}
  \mu_x &:= \sigma_{\min}(
    U_{\mathcal X}^\top(I-A_{ff})\,
    U_{\mathcal X}),\\
  \mu_y &:= \sigma_{\min}(
       U_{\mathcal Y}^\top(I-A_{ss})\,
       U_{\mathcal Y}),\\
  c_3&:= U_{\mathcal X}\bigl(
       U_{\mathcal X}^\top(I\!-\!A_{ff})\,
       U_{\mathcal X}\bigr)^{\!-1}
       U_{\mathcal X}^\top,\\
  c_4
  &:= U_{\mathcal Y}\bigl(
       U_{\mathcal Y}^\top(I\!-\!A_{ss})\,
       U_{\mathcal Y}\bigr)^{\!-1}
       U_{\mathcal Y}^\top.
  \end{aligned}
\end{align}

\vspace{-0.05in}\noindent Note that the Hurwitz condition on
$U_{\mathcal X}^\top A_{ff}\,U_{\mathcal X}$ implies that $\mu_{x} > 0$. Additionaly, we assume the following conditions. 

\begin{assumption}\label{as:resolvent_smallgain}
The restricted resolvent constant $\mu_{y} > 0$. Moreover, the restricted resolvent matrices $(I - c_3 A_{fs} c_4 A_{sf})$ and
$(I - c_4 A_{sf} c_3 A_{fs})$ are invertible.
\end{assumption}


\begin{assumption}
  \label{as:md}
$\mathbb E[\varepsilon_t\mid\mathcal F_t]=0$,
$\mathbb E[\psi_t\mid\mathcal F_t]=0$ a.s., and
$\mathbb E[\|\varepsilon_t\|^2\mid\mathcal F_t]
  \le C_\varepsilon$,
$\mathbb E[\|\psi_t\|^2\mid\mathcal F_t]
  \le C_\psi$,
where $\mathcal F_t$ is the $\sigma$-algebra generated
by $\{(x_k,y_k)\}_{k\le t}$.
\end{assumption}

\noindent Assumption~\ref{as:md} is standard in the
SA literature
\cite{srikant2019finite,MouLiWaiBarJor20,Kaledin20a}.
\vspace{-0.1in}
\section{Main Result}\label{sec:main}
\vspace{-0.05in}Recall that $(x^*,y^*)$ are the solution of the unconstrained system in ~\eqref{eq:root} and  $(x_p^*,y_p^*)\in\mathcal X\times\mathcal Y$ are the fixed point solution of the constrained updates in \eqref{eq:alg}. Existence and uniqueness of these points follow from Assumptions \ref{as:block_stab} and \ref{as:proj_inv}. As illustrated in Fig. \ref{fig:decomposition}, the iterates generated by~\eqref{eq:alg} can converge at best to $(x_p^*,y_p^*)$. Our main theorem below formalizes this behavior by establishing finite-time convergence rate bounds for the PR averages $\bar{x}_{T}$ and $\bar{y}_{T}$. Specifically, these bounds consist of two components: a \textbf{statistical error}, characterizing the rate at which the PR averages converge to $(x_p^*,y_p^*)$, and an \textbf{approximation error}, representing the distance between the unconstrained and constrained solutions

\vspace{-0.05in}\noindent 
\begin{align*}
  \epsilon_x^2
  := \inf_{x\in\mathcal X}\|x^*-x\|^2\quad \text{ and }\quad  \epsilon_y^2
  := \inf_{y\in\mathcal Y}\|y^*-y\|^2.
\end{align*}

\vspace{-0.05in}\noindent For convenience, we consider the following notation
\vspace{-0.05in}
\begin{align*}
    &m_x := \sigma_{\min}(
    U_{\mathcal X}^\top A_{ff}\,
    U_{\mathcal X}),\\
    &m_y := \sigma_{\min}(
    U_{\mathcal Y}^\top
    (A_{ss}+A_{sf}\,c_1\,A_{fs})\,
    U_{\mathcal Y}),\\
    &\kappa_{xy}
  := \|(I - c_3 A_{fs} c_4 A_{sf})^{-1}\|_2,\\
  &\kappa_{yx}
  := \|(I - c_4 A_{sf} c_3 A_{fs})^{-1}\|_2.
\end{align*}

\vspace{-0.05in}
\begin{thm}\label{thm:main}
Suppose that Assumptions~\ref{as:block_stab}--\ref{as:md}
hold.
Let $\alpha,\beta>0$ be constant step sizes with
$\beta/\alpha\ll 1$,
$\alpha < 1/\|A_{ff}\|_2$, and
$\beta < 1/\|A_{ss}\|_2$.
Then we have\vspace{-0.05in}
\begin{align}  
\begin{aligned}
  \mathbb{E}\!\bigl[\|\bar x_T\!-\!x^*\|^2\bigr]
  &\le 2B_{xx}\,\epsilon_x^2
       + 2B_{xy}\,\epsilon_y^2
       + \frac{2L_x}{T},
  \\
  \mathbb{E}\!\bigl[\|\bar y_T\!-\!y^*\|^2\bigr]
  &\le 2B_{yx}\,\epsilon_x^2
       + 2B_{yy}\,\epsilon_y^2
       + \frac{2L_y}{T},
       \end{aligned}
\end{align}

\vspace{-0.05in}\noindent where $L_x,L_y$ are the statistical constants
\vspace{-0.05in}
\begin{align*}
  L_y
  &= \frac{2\,C_\psi}{m_y^2}
     + \frac{2\|A_{sf}\|_2^2}{m_y^2\,m_x^2}\,
       C_\varepsilon,\\
  L_x
  &= \frac{2\|A_{fs}\|_2^2}{m_x^2\,m_y^2}\,C_\psi
     + \!\left(\frac{4\|A_{fs}\|_2^2\|A_{sf}\|_2^2}
            {m_y^2\,m_x^4}
       + \frac{1}{m_x^2}
     \right)\!C_\varepsilon,
\end{align*}
\vspace{-0.05in}\noindent and the approximation constants defined as
\vspace{-0.05in}
\begin{align*}
  B_{xx}
  &= 2\kappa_{xy}^2
     \!\left(1+\frac{\|A_{ff}\|_2}{\mu_x}\right)^{\!2},
  \\
  B_{xy}
  &= 2\kappa_{xy}^2
     \!\left(\frac{\|A_{fs}\|_2}{\mu_x}\right)^{\!2}
     \!\left(1+\frac{\|A_{ss}\|_2}{\mu_y}\right)^{\!2},
  \\
  B_{yy}
  &= 2\kappa_{yx}^2
     \!\left(1+\frac{\|A_{ss}\|_2}{\mu_y}\right)^{\!2},
  \\
  B_{yx}
  &= 2\kappa_{yx}^2
     \!\left(\frac{\|A_{sf}\|_2}{\mu_y}\right)^{\!2}
     \!\left(1+\frac{\|A_{ff}\|_2}{\mu_x}\right)^{\!2}.
\end{align*}\vspace{0.05cm}
\end{thm} 
\begin{remark}
The bound in Theorem~\ref{thm:main} decomposes the optimal errors into two distinct components. The \textbf{statistical constants} $L_x,L_y$ govern the
convergence of the PR averages toward the
constrained solution $(x_p^*,y_p^*)$ at a rate $O(1/T)$.
Their structure reveals how noise propagates through
the two-time-scale coupling:
$L_y$ depends on the slow noise variance $C_\psi$
directly and on the fast noise $C_\varepsilon$ through
the cross-coupling norm $\|A_{sf}\|_2$, both scaled
by the projected stability margins $m_x,m_y$.
Similarly, $L_x$ inherits contributions from $C_\psi$
through $\|A_{fs}\|_2$.
Smaller margins $m_x,m_y$ (i.e., weaker projected
stability) amplify the statistical error.

The \textbf{approximation constants} $B_{xx},B_{xy},
B_{yy},B_{yx}$ are
controlled by the resolvent margins $\mu_x,\mu_y$
and the coupling constants $\kappa_{xy},\kappa_{yx}$.
The off-diagonal terms $B_{xy},B_{yx}$ capture an effect unique to the TTSA setting:
even if $\epsilon_x^2$ is small, the $y$-error
inherits a contribution through $B_{yx}$, and
vice versa.

Our theoretical bound provides a clean representation on these two expected errors of the projected TTSA in \eqref{eq:alg}:
the subspace choice determines the approximation errors $\epsilon_x^2,\epsilon_y^2$, while the
statistical errors decay with the number of
iterations $T$ regardless of the subspace.
In Section~\ref{sec:simulation}, we verify this decomposition empirically: the PR-averaged errors
decay and then plateau at the approximation error, with the $O(1/T)$ rate consistent across different subspace choices. We also discuss how to choose the subspaces to control the values of $\epsilon_x^2,\epsilon_y^2$.
\end{remark}
\section{Numerical Experiments}\label{sec:simulation}\vspace{-0.05in}
In this section, we provide numerical simulations to validate the theoretical bounds in Theorem~\ref{thm:main}. We apply the projected TTSA in \eqref{eq:alg} to two settings: a synthetic coupled system and GTD for policy evaluation in reinforcement learning. 
In the tabular setting, GTD 
is an unprojected TTSA \cite{SuttonEtAl09GTD}: both the value
function $V^\pi\in\mathbb R^{|\mathcal S|}$ and the auxiliary
variable $w\in\mathbb R^{|\mathcal S|}$ are updated in
the full ambient space.
When the state space is large, one can use linear function approximation $V^\pi\approx\Phi\theta$ with the feature matrix
$\Phi\in\mathbb R^{|\mathcal S|\times d}$, $d\ll|\mathcal S|$,
which restricts both variables to the subspace
$\mathrm{span}(\Phi)$.
This is precisely the projected TTSA
in~\eqref{eq:alg}, with
$\mathcal X=\mathcal Y=\mathrm{span}(\Phi)$ and
$\Pi_{\mathcal X}=\Pi_{\mathcal Y}=\Phi(\Phi^\top\Phi)^{-1}\Phi^\top$.
\vspace{-0.1in}
\subsection{Synthetic Coupled Linear Systems}
\label{subsec:synth_exp}

\vspace{-0.05in}We construct a stable coupled system of the
form~\eqref{eq:root} with ambient dimensions
$n=20$ and $m=16$.
The self-dynamics matrices $A_{ff}$ and $A_{ss}$ are
negative definite (diagonal in random orthonormal bases
$\widetilde U_x\!\in\!\mathbb{R}^{n\times n}$,
$\widetilde U_y\!\in\!\mathbb{R}^{m\times m}$),
while the cross-couplings $A_{fs},A_{sf}$ are random
matrices scaled by $0.08$ to preserve overall stability.
A ground-truth solution $(x^*,y^*)$ is defined with
exponentially decaying coordinates in
$\widetilde U_x,\widetilde U_y$, and $(b_1,b_2)$ are
set so that \eqref{eq:root} holds.
We constrain the iterates to rank-$r$ subspaces
$\mathcal X = \mathrm{span}(\widetilde U_x(:,1\!:\!r))$
and
$\mathcal Y = \mathrm{span}(\widetilde U_y(:,1\!:\!r))$
with $r=6$, and compute the constrained solution
$(x_p^*,y_p^*)$ by solving the reduced system.
The algorithm~\eqref{eq:alg} is implemented for $T=5{,}000$
iterations, averaged over $25$ independent trials,
with i.i.d.\ Gaussian noise
$\varepsilon_t,\psi_t\sim\mathcal N(0,\sigma^2 I)$,
$\sigma=0.08$, and constant step sizes
$\alpha=0.2$ (fast) and $\beta=0.01$ (slow),
giving $\beta/\alpha=0.05$.

Figure~\ref{fig:synthetic} plots the PR-averaged errors $\|\bar x_t-x^*\|^2$ and
$\|\bar y_t-y^*\|^2$ on a log scale, with horizontal
dashed lines at the approximation errors
$\|x_p^*-x^*\|^2$ and $\|y_p^*-y^*\|^2$.
Both PR errors decrease and then plateau near the
corresponding dashed lines, confirming the decomposition in Theorem~\ref{thm:main}: the
statistical term decays as $O(1/T)$, while the limiting accuracy is set by the approximation error.
The fast iterate~$x$ ($\alpha=0.2$)
exhibits a more rapid initial decay, whereas the slow iterate~$y$ ($\beta=0.01$) shows a longer
transient, consistent with the two-time-scale structure.
To further validate the $O(1/T)$ rate, we overlay the
curve $L_y/T + \|y_p^*-y^*\|^2$ (dotted), where $L_y$
is a fitted constant.
This curve lies above the slow PR error
$\|\bar y_t - y^*\|^2$ for all~$T$, confirming that
the bound in Theorem~\ref{thm:main} correctly captures
both the $O(1/T)$ decay rate and the approximation
error floor.
\vspace{-0.1in}
\begin{figure}[htb]
  \centering
  \includegraphics[width=\linewidth]%
    {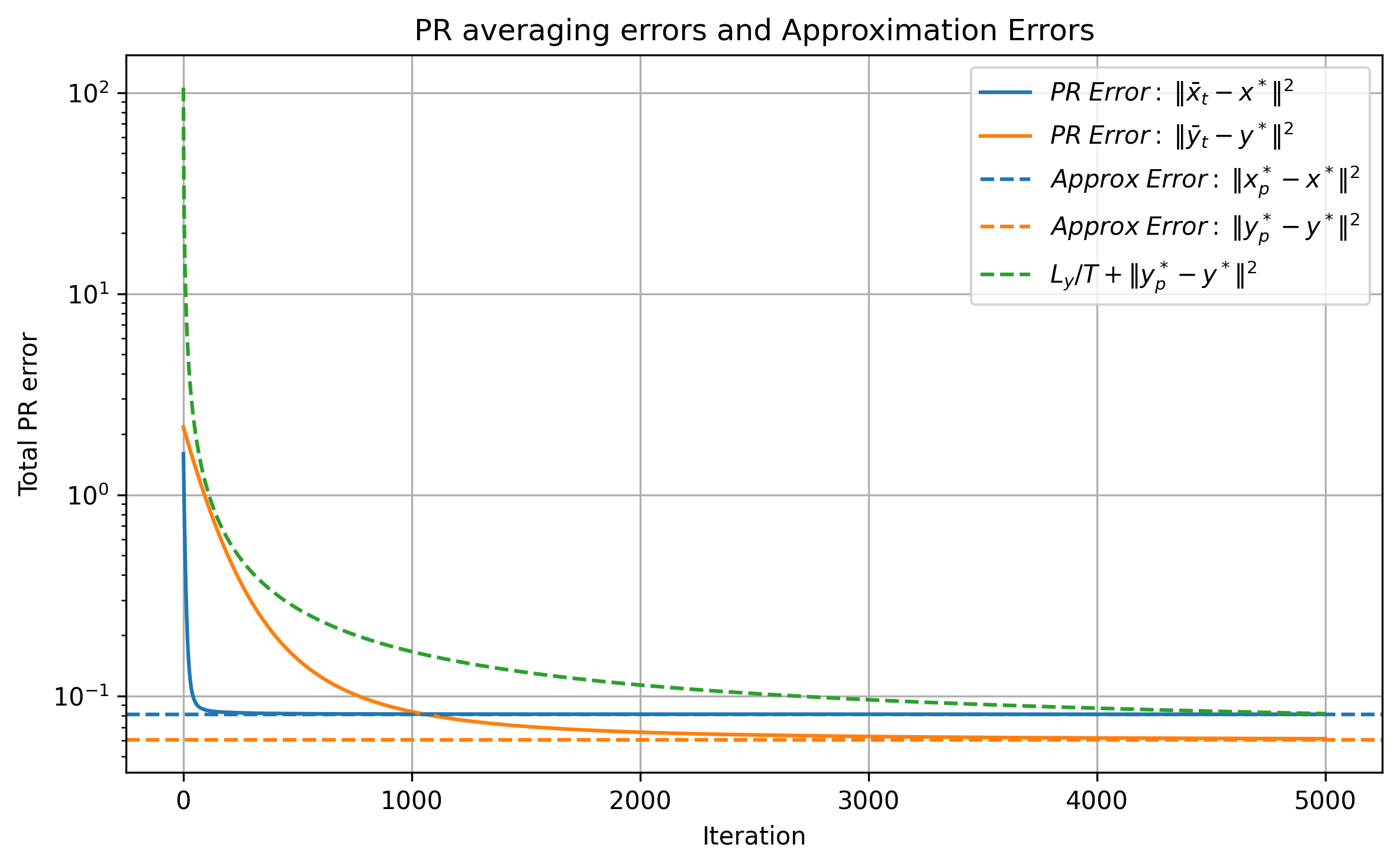}
  \vspace{-0.2in}
  \caption{%
    Synthetic coupled system.}\vspace{-0.2in}
  \label{fig:synthetic}
\end{figure}

\subsection{GTD with Feature Mismatch}\vspace{-0.05in}
\label{subsec:gtd_exp}
In this section, we apply the projected TTSA to study the GTD algorithm. Using the simulation of GTD, our first goal is to verify that the bias--variance decomposition in Theorem~\ref{thm:main} also holds in a reinforcement learning setting.  
Second, we investigate how the choice of feature subspace $\mathrm{span}(\Phi)$ affects each component
of the theoretical bound---in particular, whether the $O(1/T)$ convergence rate of the statistical error is robust across subspaces with very different alignment to the value function.

We consider policy evaluation on a $9$-state Markov decision process (MDP) with a discount factor $\gamma=0.9$.
The discounted value function $V^\pi\in\mathbb{R}^9$ is the
unique solution of the Bellman equation of the underlying MDP. 
We will consider the linear function approximation setting where $V^\pi\approx\Phi\theta$ within a low-dimensional feature subspace $\mathrm{span}(\Phi)$. The goal of GTD is to find $\theta^*$ such that $\Phi\theta^*$ is the best approximation of $V^{\pi}$.

To evaluate the effects of different choices of the subspace spanned by features $\Phi$ on the convergence of GTD, we consider a random orthonormal basis
$W=[w_1,\dots,w_9]\in\mathbb{R}^{9\times 9}$
and set $V^{\pi} = Wc$ with coefficients
\vspace{-0.1in}
\[
  c = [3.0,\,2.4,\,1.8,\,1.4,\,1.1,\,0.5,\,10^{-3},\,5\!\cdot\!10^{-4},\,10^{-4}]^\top\!.\vspace{-0.1in}
\]
By design, $V^\pi$ concentrates nearly all of its
energy in the first six directions
$w_1,\dots,w_6$, while the last three directions
$w_7,w_8,w_9$ carry negligible energy.
This construction will enable us to study the alignment between the feature subspace and the value function, and thereby illustrate different aspects of the approximation error $\epsilon^2$ in the bound in Theorem \ref{thm:main}.

We will approximate
$V^\pi$ by $\Phi\theta$ with $\Phi\in\mathbb{R}^{9\times 3}$. Given a choice of $\Phi$, we will implement GTD to find $\theta^*$ that gives the best approximation of $V^{\pi}$. As mentioned, GTD can be formulated as a variant of the projected TTSA in \eqref{eq:alg}, where $\mathcal X=\mathcal Y=\mathrm{span}(\Phi)$. For our simulations, we set step sizes $\alpha=0.02$ and $\beta=0.001$
($\beta/\alpha=0.05$) and conduct experiments in three different choices of feature matrices $\Phi$:
(i)~\textbf{well aligned}:
    $\Phi^1=[w_1,w_2,w_3]$, spanning the
    highest-energy directions of $V^\pi$;
(ii)~\textbf{medium aligned}:
    $\Phi^2$ is a random orthonormal basis in
    $\mathbb{R}^{9\times 3}$ (via QR factorization);
(iii)~\textbf{poorly aligned}:
    $\Phi^3=[w_7,w_8,w_9]$, spanning the
    lowest-energy directions, nearly orthogonal to
    $V^\pi$. For each $\Phi^i$, we report the the approximation error
$\|V^\pi-\Phi^i\theta^{i,*}\|^2$, 
the statistical error
$\|\Phi^i\theta^{i,*}
  -\Phi^i\bar\theta^i_t\|^2$, and total PR error
$\|V^\pi-\Phi^i\bar\theta^i_t\|^2$. Our simulation results are shown in Figure \ref{fig:gtd}.

The top plot in Figure~\ref{fig:gtd} displays the approximation
error $\|V^\pi-\Phi^i\theta^{i,*}\|^2$ across the three feature families.
As expected, well-aligned features yield the smallest
approximation error. On the other hand, the middle plot in Figure~\ref{fig:gtd} displays the statistical
error $\|\Phi^i\theta^{i,*}-\Phi^i\bar\theta^i_t\|^2$. We order the curves from well-aligned (highest) to poorly-aligned (lowest). This ordering reflects the initial error magnitude, not the convergence rate: since $V^\pi$ has positive
entries and well-aligned features capture more of its
energy, the projected value
$\Phi^1\theta^{1,*}$ is largest in norm, producing a larger initial gap
$\|\Phi^i\theta^{i,*}-\Phi^i\bar\theta^i_0\|^2$.
Crucially, the slopes of all three curves are
similar, indicating a consistent decaying rate 
across feature choices.

Finally, the bottom plot in Figure~\ref{fig:gtd} shows the total PR error
$\|V^\pi-\Phi^i\bar\theta^i_t\|^2$ for the
well-aligned case $\Phi^1$, together with its
statistical and approximation components on the same
log-scale axis. The statistical error decays to zero while the total
PR error plateaus at the approximation error as shown by our theoretical bound in 
Theorem~\ref{thm:main}.
We also overlay the curve $L_y/T$ (dotted), where
$L_y$ is a fitted constant, which lies above the
statistical error for all~$T$.
This confirms that the statistical error decays at
the $O(1/T)$ rate predicted by Theorem~\ref{thm:main}.
\vspace{-0.1in}\noindent 

\begin{figure}[htb]
  \centering
  \includegraphics[width=\linewidth]%
    {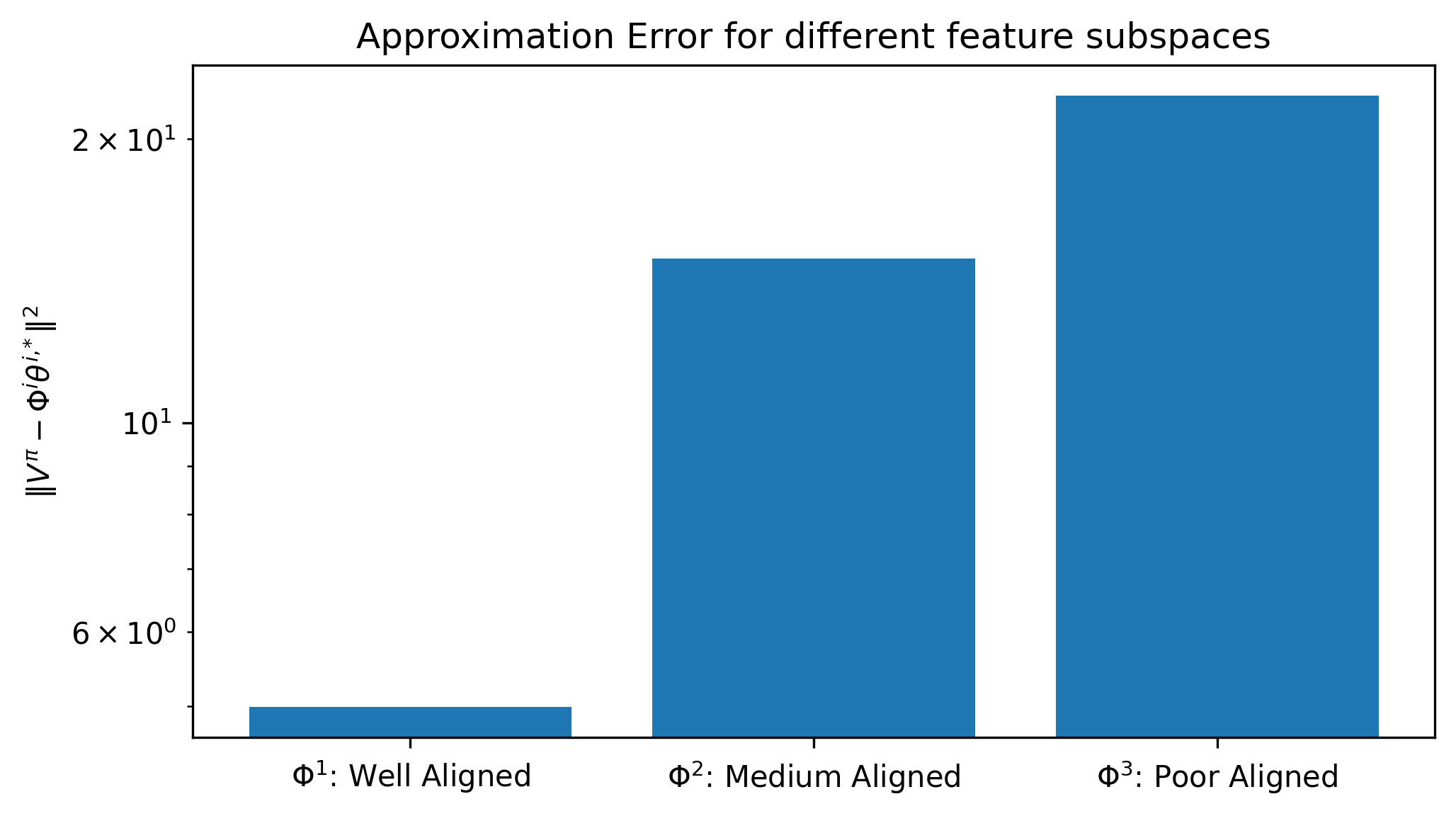}\\[2pt]
  \includegraphics[width=\linewidth]%
    {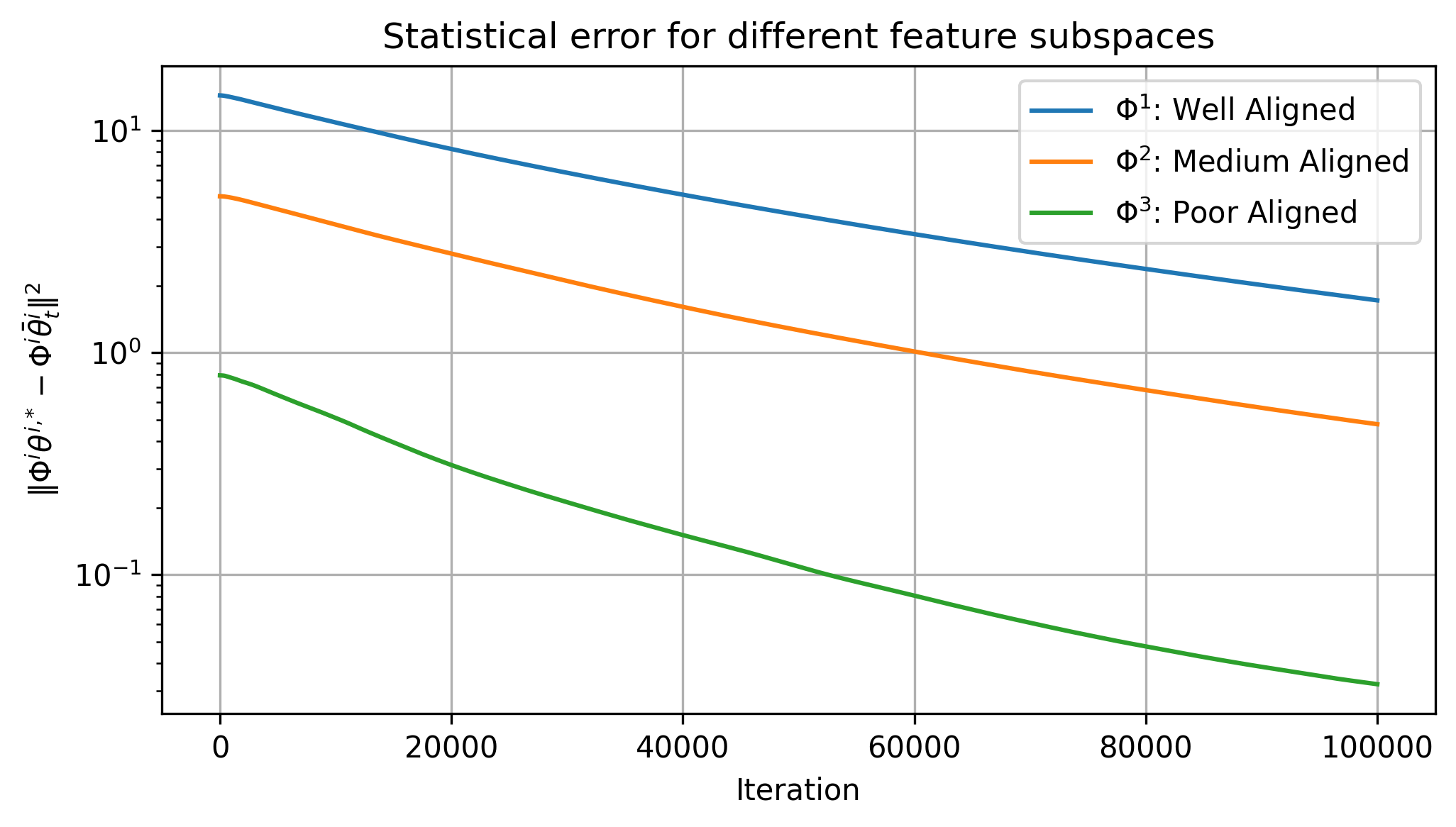}\\[2pt]
  \includegraphics[width=\linewidth]%
    {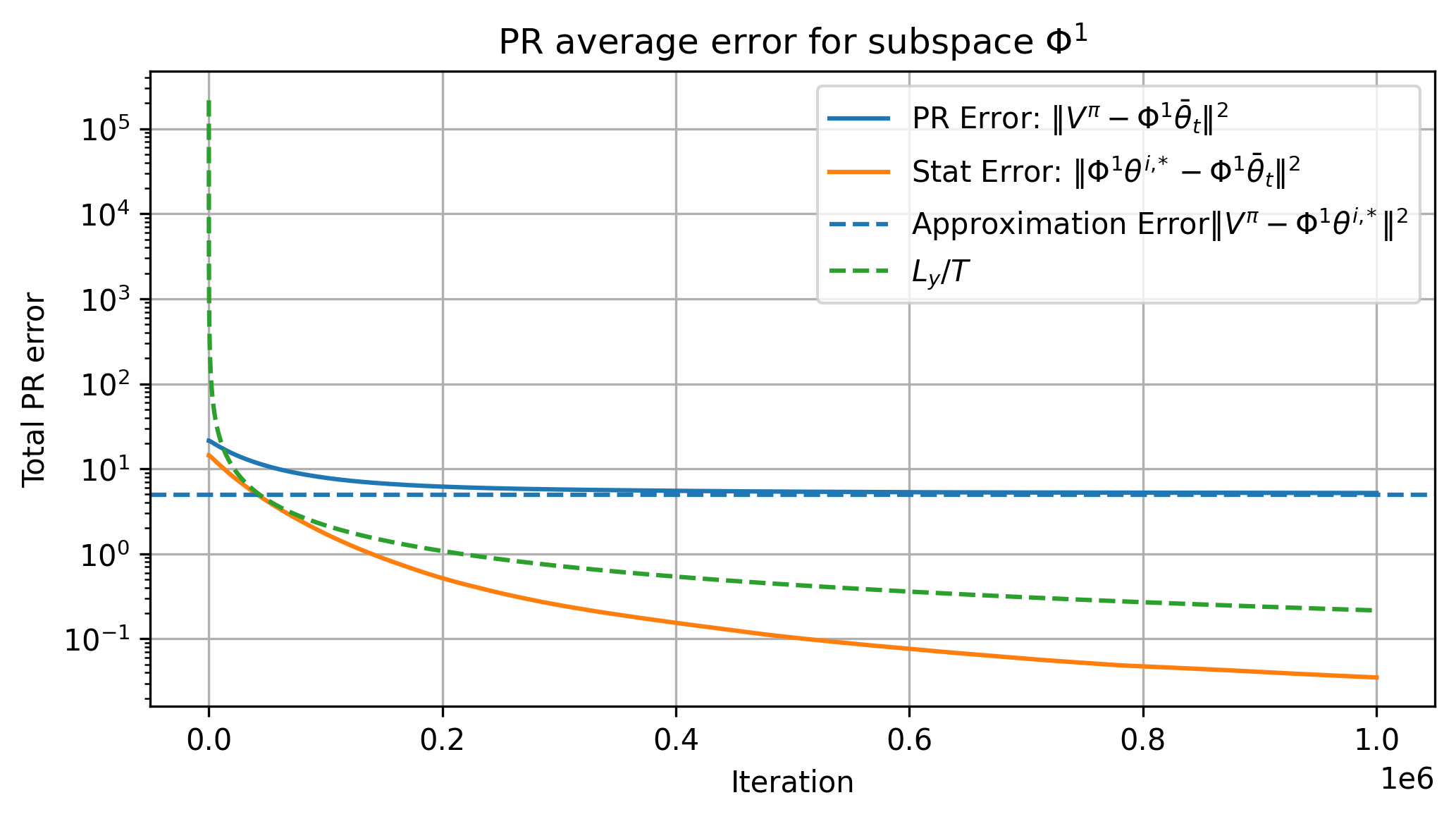}
  \vspace{-0.25in}
  \caption{Performance of GTD under different $\Phi$.}
  \label{fig:gtd}
\end{figure}







\vspace{-0.05in}
\section*{APPENDIX}\vspace{-0.05in}
\section{Proof of Theorem~\ref{thm:main}}\vspace{-0.05in}
\label{sec:appendix}

We decompose the PR-averaged error into
(i)~a \emph{statistical term} measuring how well
$\bar x_T,\bar y_T$ track the constrained solution
$(x_p^*,y_p^*)$, and
(ii)~an \emph{approximation term} measuring the
distance between $(x_p^*,y_p^*)$ and the unconstrained
solution $(x^*,y^*)$.
Throughout, we use the fact that the projections are
linear and that the iterates remain feasible:
$x_t\in\mathcal X$ and $y_t\in\mathcal Y$ for
all~$t$. Our analysis is composed of $5$ steps.

\noindent\textbf{Step~0: Projection-solution identity.}
Solving a projected linear equation over a subspace
reduces to an $r$-dimensional system in reduced
coordinates.

\begin{lemma}[Projected linear solve]
  \label{lem:projection_solve}
Let $\mathcal Z$ be a linear subspace with orthonormal
basis $U_{\mathcal Z}$, and let $A$ be a matrix such
that $U_{\mathcal Z}^\top A\,U_{\mathcal Z}$ is
invertible.
If $z_p\in\mathcal Z$ satisfies
$\Pi_{\mathcal Z}(Az_p)=\Pi_{\mathcal Z}(b)$, then
\vspace{-0.1in}
\[
  z_p
  = U_{\mathcal Z}
    (U_{\mathcal Z}^\top A\,U_{\mathcal Z})^{-1}
    U_{\mathcal Z}^\top b.
\]
\end{lemma}

\begin{proof}
Write $z_p=U_{\mathcal Z}\tilde z$ and left-multiply $\Pi_{\mathcal Z}(Az_p)=\Pi_{\mathcal Z}(b)$ by $U_{\mathcal Z}^\top$ to obtain $U_{\mathcal Z}^\top A\,U_{\mathcal Z}\tilde z = U_{\mathcal Z}^\top b$; inverting yields the claim.
\end{proof}

\noindent\textbf{Step~1: Bias--variance decomposition.}
Using $\|a+b\|^2\le 2\|a\|^2+2\|b\|^2$ with
$a=\bar x_T-x_p^*$ and $b=x_p^*-x^*$ gives\vspace{-0.05in}
\begin{align}\label{eq:decomp}
  \begin{aligned}
  \mathbb E\|\bar x_T\!-\!x^*\|^2
  &\le 2\mathbb E\|\bar x_T\!-\!x_p^*\|^2
       + 2\|x_p^*\!-\!x^*\|^2,\\
  \mathbb E\|\bar y_T\!-\!y^*\|^2
  &\le 2\mathbb E\|\bar y_T\!-\!y_p^*\|^2
       + 2\|y_p^*\!-\!y^*\|^2.
  \end{aligned}
\end{align}

\vspace{-0.05in}\noindent It remains to bound (i)~the statistical errors
$\mathbb E\|\bar x_T-x_p^*\|^2$ and
$\mathbb E\|\bar y_T-y_p^*\|^2$ (\textbf{Part~1}),
and (ii)~the approximation errors
$\|x_p^*-x^*\|^2$ and $\|y_p^*-y^*\|^2$
(\textbf{Part~2}).

\noindent\textbf{Part~1: Statistical error.}
We bound
$\mathbb E[\|\bar y_T-y_p^*\|^2]$ and
$\mathbb E[\|\bar x_T-x_p^*\|^2]$ by relating the
block averages to the constrained solution via telescoping
sums.
The key ideas are:
(i)~define the instantaneous constrained solution
$x_p^*(y)$ at any fixed~$y$;
(ii)~express $\bar x_T-x_p^*(\bar y_T)$ and
$\bar y_T-y_p^*$ via sums of the update increments;
(iii)~invert the projected operators using
Lemma~\ref{lem:projection_solve}.

\noindent\textbf{1.\;Define $x_p^*(y)$ and the
constrained solution.}
For each $y\in\mathbb R^{m}$, let
$x_p^*(y)\in\mathcal X$ be the unique solution of\vspace{-0.05in}\noindent 
\[
  \Pi_{\mathcal X}\!\bigl(
    A_{ff}x_p^*(y)+A_{fs}y-b_1
  \bigr)=0.
\]

\vspace{-0.05in}\noindent By Lemma~\ref{lem:projection_solve} with
$A=A_{ff}$ and $b=b_1-A_{fs}y$,
\begin{equation}\label{eq:xpstar_y}
  x_p^*(y)
  = c_1(b_1-A_{fs}y),
\end{equation}

\vspace{-0.05in}\noindent where
$c_1 := U_{\mathcal X}
  (U_{\mathcal X}^\top A_{ff}\,U_{\mathcal X})^{-1}
  U_{\mathcal X}^\top$.
Note that $x_p^*=x_p^*(y_p^*)$ by definition.

\noindent\textbf{2.\;Telescoping sums and block
averages.}
Fix $T\ge 1$ and define the block averages and noise
sums

\vspace{-0.05in}\noindent 
\begin{align*}
  \bar x_T
  &= \tfrac{1}{T}
     \textstyle\sum_{t=0}^{T-1} x_t,
  &
  \bar y_T
  &= \tfrac{1}{T}
     \textstyle\sum_{t=0}^{T-1} y_t,\\
  \bar\varepsilon_T
  &= \textstyle\sum_{t=0}^{T-1}\varepsilon_t,
  &
  \bar\psi_T
  &= \textstyle\sum_{t=0}^{T-1}\psi_t.
\end{align*}

\vspace{-0.05in}\noindent \emph{$y$-update.}
Summing the $y$-recursion~\eqref{eq:alg} from $t=0$
to $T-1$ with constant step size $\beta$, expanding
in block averages, and centering around the
constrained-solution condition
$\Pi_{\mathcal Y}b_2
  = \Pi_{\mathcal Y}(A_{sf}x_p^* + A_{ss}y_p^*)$
gives

\vspace{-0.05in}\noindent \begin{align}
  &\Pi_{\mathcal Y}A_{ss}(\bar y_T - y_p^*)
   + \Pi_{\mathcal Y}A_{sf}(\bar x_T - x_p^*)
  \notag\\
  &\qquad= \frac{1}{T}\Bigl(
       \frac{y_0 - y_T}{\beta}
       - \Pi_{\mathcal Y}\bar\psi_T
     \Bigr).
  \label{eq:y_centered}
\end{align}

\vspace{-0.05in}\noindent \emph{$x$-update.}
An analogous telescoping argument for the
$x$-recursion~\eqref{eq:alg} with constant step
size~$\alpha$ gives

\vspace{-0.05in}\noindent 
\begin{align}
  &\Pi_{\mathcal X}A_{ff}
   \bigl(\bar x_T - x_p^*(\bar y_T)\bigr)
   = \frac{1}{T}\Bigl(
       \frac{x_0 - x_T}{\alpha}
       - \Pi_{\mathcal X}\bar\varepsilon_T
     \Bigr),
  \label{eq:x_centered}
\end{align}

\vspace{-0.05in}\noindent where we used the defining property
$\Pi_{\mathcal X}b_1
  = \Pi_{\mathcal X}(A_{ff}x_p^*(\bar y_T)
    + A_{fs}\bar y_T)$.

\noindent\textbf{3.\;Invert the projected operators.}
Since
$\bar x_T - x_p^*(\bar y_T)\in\mathcal X$,
Lemma~\ref{lem:projection_solve} applied
to~\eqref{eq:x_centered} yields

\vspace{-0.05in}\noindent 
\begin{align}
  \bar x_T - x_p^*(\bar y_T)
  &= \frac{c_1}{T}\Bigl(
       \frac{x_0 - x_T}{\alpha}
       - \Pi_{\mathcal X}\bar\varepsilon_T
     \Bigr).
  \label{eq:x_inverted}
\end{align}

\vspace{-0.05in}\noindent Plugging the decomposition
$\bar x_T - x_p^*
  = (\bar x_T - x_p^*(\bar y_T))
    + (x_p^*(\bar y_T) - x_p^*)$
into~\eqref{eq:y_centered}, using
$x_p^*(\bar y_T) - x_p^*
  = c_1 A_{fs}(\bar y_T - y_p^*)$
from~\eqref{eq:xpstar_y}, and rearranging:
\vspace{-0.1in}
\begin{align*}
  \Pi_{\mathcal Y}(A_{ss} + A_{sf}c_1A_{fs})
   (\bar y_T - y_p^*)&= \frac{1}{T}\Bigl(
       \frac{y_0 - y_T}{\beta}
       - \Pi_{\mathcal Y}\bar\psi_T
     \Bigr)\\
  &\quad
     - \Pi_{\mathcal Y}A_{sf}
       (\bar x_T - x_p^*(\bar y_T)).\vspace{-0.05in}
\end{align*}

\vspace{-0.05in}\noindent Applying Lemma~\ref{lem:projection_solve} with
$A = A_{ss} + A_{sf}c_1A_{fs}$
(invertible by Assumption~\ref{as:proj_inv}) and
$c_2 := U_{\mathcal Y}
  (U_{\mathcal Y}^\top
   (A_{ss} + A_{sf}c_1A_{fs})
   U_{\mathcal Y})^{-1}
  U_{\mathcal Y}^\top$
gives
\begin{align}
  \bar y_T - y_p^*
  &= \underbrace{
       \frac{c_2}{T}\Bigl(
         \frac{y_0 - y_T}{\beta}
         - \Pi_{\mathcal Y}\bar\psi_T
       \Bigr)
     }_{=:\,T_1}
  \notag\\
  &\quad
     \underbrace{
       -\,c_2\Pi_{\mathcal Y}A_{sf}
       \frac{c_1}{T}\Bigl(
         \frac{x_0 - x_T}{\alpha}
         - \Pi_{\mathcal X}\bar\varepsilon_T
       \Bigr)
     }_{=:\,T_2},
  \label{eq:y_final}
\end{align}

\vspace{-0.05in}\noindent where we substituted~\eqref{eq:x_inverted} for
$\bar x_T - x_p^*(\bar y_T)$.

\noindent\textbf{4.\;Take norms and expectations for
the $y$-bound.}
From~\eqref{eq:y_final} and
$\|a+b\|^2\le 2(\|a\|^2+\|b\|^2)$:
\begin{align}
  \mathbb E[\|\bar y_T - y_p^*\|^2]
  &\le 2\,\mathbb E[\|T_1\|^2]
       + 2\,\mathbb E[\|T_2\|^2].
  \label{eq:y_two_terms}
\end{align}

\vspace{-0.05in}\emph{Bounding $\mathbb E[\|T_1\|^2]$.}
Expanding $T_1$ and using the fact that the cross term
$\langle c_2(y_0\!-\!y_n)/\beta,\;
  c_2\Pi_{\mathcal Y}\bar\psi_T\rangle$
has zero expectation (by the martingale-difference
property, Assumption~\ref{as:md}):
\begin{align}
  \mathbb E[\|T_1\|^2]
  &\le \frac{1}{T^2}
       \mathbb E\!\Big[
         \Big\|c_2\frac{y_0 - y_T}{\beta}\Big\|^2
       \Big]
       + \frac{1}{T^2}
       \mathbb E\!\bigl[
         \|c_2\Pi_{\mathcal Y}\bar\psi_T\|^2
       \bigr].
  \label{eq:T1_bound}
\end{align}
\vspace{-0.25in}

\noindent By the martingale-difference property,
\vspace{-0.05in}
\begin{align*}
  \mathbb E[\|c_2\Pi_{\mathcal Y}\bar\psi_T\|^2]
  &= \sum_{t=0}^{T-1}
     \mathbb E[\|c_2\Pi_{\mathcal Y}\psi_t\|^2]
     \le \|c_2\|_2^2\,T\,C_\psi.
\end{align*}

\vspace{-0.05in}\noindent Hence the noise contribution in~\eqref{eq:T1_bound}
is $\|c_2\|_2^2\,C_\psi/T = O(1/T)$, while the
boundary term is $O(1/T^2)$ since the numerator
$\mathbb E[\|c_2(y_0-y_n)/\beta\|^2]$ remains
bounded: the projected stability guaranteed by
Assumption~\ref{as:proj_inv} and the step-size
conditions ensure uniformly bounded iterates
\cite{srikant2019finite}.

\emph{Bounding $\mathbb E[\|T_2\|^2]$.}
An identical argument (replacing
$c_2,\,\bar\psi_T,\,y,\,\beta$ by
$c_2\Pi_{\mathcal Y}A_{sf}c_1,\,
  \bar\varepsilon_T,\,x,\,\alpha$)
gives
\vspace{-0.05in}
\begin{align}
  \mathbb E[\|T_2\|^2]
  &\le \frac{1}{T^2}
       \mathbb E\!\Big[
         \Big\|c_2\Pi_{\mathcal Y}A_{sf}c_1
           \frac{x_0 - x_T}{\alpha}\Big\|^2
       \Big]
  \notag\\
  &\quad
       + \frac{\|c_2\Pi_{\mathcal Y}A_{sf}
               c_1\|_2^2\,C_\varepsilon}{T}.
  \label{eq:T2_bound}
\end{align}

\vspace{-0.05in}\noindent \emph{Assembling the $y$-bound.}
Substituting~\eqref{eq:T1_bound}--\eqref{eq:T2_bound}
into~\eqref{eq:y_two_terms} and separating the
$O(1/T)$ and $O(1/T^2)$ contributions:
\vspace{-0.1in}
\begin{align}
  \mathbb E[\|\bar y_T - y_p^*\|^2]
  &\le \frac{L_y}{T} + \frac{D_y}{T^2},
  \label{eq:y_stat_full}
\end{align}

\vspace{-0.05in}\noindent where the leading statistical constant is
\vspace{-0.05in}
\begin{align*}
  L_y
  &= 2\|c_2\|_2^2\,C_\psi
     + 2\|c_2\Pi_{\mathcal Y}A_{sf}
         c_1\|_2^2\,C_\varepsilon\\
  &\le \frac{2\,C_\psi}{m_y^2}
       + \frac{2\|A_{sf}\|_2^2}{m_y^2\,m_x^2}\,C_\varepsilon, \vspace{-0.2in}
\end{align*}\vspace{-0.15in}

The remainder $D_y$ collects $O(1/T^2)$ boundary terms
involving $\mathbb E[\|y_0-y_T\|^2]$ and
$\mathbb E[\|x_0-x_T\|^2]$, which are finite since
the step-size conditions and
Assumption~\ref{as:proj_inv} ensure uniformly bounded
iterates \cite{srikant2019finite}.

\noindent\textbf{5.\;Bound
$\mathbb E\|\bar x_T - x_p^*\|^2$.}
An analogous argument using
$\Pi_{\mathcal X}b_1
  = \Pi_{\mathcal X}(A_{ff}x_p^* + A_{fs}y_p^*)$
and Lemma~\ref{lem:projection_solve} gives
\vspace{-0.05in}\noindent 
\begin{align*}
  \bar x_T - x_p^*
  &= \frac{c_1}{T}\Bigl(
       \frac{x_0 - x_T}{\alpha}
       - \Pi_{\mathcal X}\bar\varepsilon_T
     \Bigr)
     - c_1\Pi_{\mathcal X}A_{fs}
       (\bar y_T - y_p^*).
\end{align*}

\vspace{-0.05in}\noindent Applying $\|a+b\|^2\le 2\|a\|^2+2\|b\|^2$,
using $\mathbb E[\|\bar\varepsilon_T\|^2]
  \le T\,C_\varepsilon$,
and substituting~\eqref{eq:y_stat_full} for the
coupling term yields
$\mathbb E[\|\bar x_T - x_p^*\|^2]
  \le L_x/T + D_x/T^2$,
where
\vspace{-0.1in}
\begin{align*}
  L_x
  &\le \frac{2\|A_{fs}\|_2^2}
            {m_x^2\,m_y^2}\,C_\psi
     + \left(
       \frac{4\|A_{fs}\|_2^2\|A_{sf}\|_2^2}
            {m_y^2\,m_x^4}
       + \frac{1}{m_x^2}
     \right)C_\varepsilon,
\end{align*}

\vspace{-0.05in}\noindent and $D_x$ collects the $O(1/T^2)$ terms.
This completes Part~1.

\noindent\textbf{Part~2: Approximation error.}
We compare the constrained solution $(x_p^*,y_p^*)$ to the
unconstrained solution~$(x^*,y^*)$.

\noindent\textbf{1.\;Decompose into projection mismatch
and coupling.}
\vspace{-0.1in}
\begin{align*}
  x_p^*\!-\!x^*
  &= (x_p^*\!-\!\Pi_{\mathcal X}x^*)
     + (\Pi_{\mathcal X}x^*\!-\!x^*),\\
  y_p^*\!-\!y^*
  &= (y_p^*\!-\!\Pi_{\mathcal Y}y^*)
     + (\Pi_{\mathcal Y}y^*\!-\!y^*),
\end{align*}

\vspace{-0.05in}\noindent where $\Pi_{\mathcal X}x^*-x^*$ and
$\Pi_{\mathcal Y}y^*-y^*$ are the best-approximation
errors onto the feasible subspaces.

\noindent\textbf{2.\;Solve the projected error equations
and eliminate cross-terms.}
Subtracting the projected $x$-equation evaluated at
$(x^*,y^*)$ from that at $(x_p^*,y_p^*)$, applying
Lemma~\ref{lem:projection_solve} with $A=I-A_{ff}$
($c_3$ as in~\eqref{eq:c3c4}), and combining with the
best-approximation error gives
$x_p^*-x^*
  = (I+c_3A_{ff})(\Pi_{\mathcal X}x^*-x^*)
    + c_3A_{fs}(y_p^*-y^*)$.
A symmetric argument with $c_4$ from~\eqref{eq:c3c4}
gives
$y_p^*-y^*
  = (I+c_4A_{ss})(\Pi_{\mathcal Y}y^*-y^*)
    + c_4A_{sf}(x_p^*-x^*)$.
Substituting the second into the first and rearranging:
\vspace{-0.05in}
\begin{align*}
  &x_p^*-x^*
  = (I\!-\!c_3A_{fs}c_4A_{sf})^{-1}
     (I\!+\!c_3A_{ff})
     (\Pi_{\mathcal X}x^*\!-\!x^*)\\
  &\qquad
     + (I\!-\!c_3A_{fs}c_4A_{sf})^{-1}
     c_3A_{fs}(I\!+\!c_4A_{ss})
     (\Pi_{\mathcal Y}y^*\!-\!y^*).
\end{align*}
\vspace{-0.2in}

\noindent Taking norms and using
$\|a+b\|^2\le 2\|a\|^2+2\|b\|^2$,
\vspace{-0.1in}
\begin{align*}
  &\|x_p^*-x^*\|^2
  \le 2\kappa_{xy}^2
     \!\left(1+\frac{\|A_{ff}\|_2}{\mu_x}\right)^{\!2}
     \|\Pi_{\mathcal X}x^*\!-\!x^*\|^2\\
  &\qquad
     + 2\kappa_{xy}^2
     \!\left(\frac{\|A_{fs}\|_2}{\mu_x}\right)^{\!2}
     \!\left(1+\frac{\|A_{ss}\|_2}{\mu_y}\right)^{\!2}
     \|\Pi_{\mathcal Y}y^*\!-\!y^*\|^2,
\end{align*}
\vspace{-0.15in}

\noindent which identifies $B_{xx}$ and $B_{xy}$.
Analogously, eliminating $x_p^*-x^*$ gives
\vspace{-0.05in}
\begin{align*}
  &\|y_p^*-y^*\|^2
  \le 2\kappa_{yx}^2
     \!\left(1+\frac{\|A_{ss}\|_2}{\mu_y}\right)^{\!2}
     \|\Pi_{\mathcal Y}y^*\!-\!y^*\|^2\\
  &\qquad
     + 2\kappa_{yx}^2
     \!\left(\frac{\|A_{sf}\|_2}{\mu_y}\right)^{\!2}
     \!\left(1+\frac{\|A_{ff}\|_2}{\mu_x}\right)^{\!2}
     \|\Pi_{\mathcal X}x^*\!-\!x^*\|^2,
\end{align*}

\vspace{-0.05in}\noindent which identifies $B_{yx}$ and $B_{yy}$.

\noindent\textbf{Conclusion.}
Combining Part~1 (statistical terms) and Part~2
(approximation terms) with the
decomposition~\eqref{eq:decomp}
from Step~1 yields the claimed bounds in
Theorem~\ref{thm:main}.
\hfill$\square$
\vspace{-0.05in}





\vspace{-0.05in}
\bibliographystyle{IEEEtran}
\bibliography{references}

\end{document}